\documentclass[letterpaper, 10 pt, conference]{ieeeconf}  

\IEEEoverridecommandlockouts                              

\usepackage{graphicx}      
\usepackage{color}
\usepackage{xcolor}
\usepackage{amssymb,amsmath,color,flushend,mathtools}
\usepackage{caption}
\usepackage{subcaption}
\usepackage[utf8]{inputenc}
\usepackage{ragged2e}
\usepackage{soul}
\usepackage{dsfont}
\usepackage{comment}
\usepackage{algorithm, algorithmic}
\usepackage{amsfonts}

\usepackage{amsthm}
\newcommand{\Ball}{\mathbb{B}}
\newcommand{\ZZ}{\mathbb{Z}}
\renewcommand{\Re}{\mathbb{R}}
\newcommand{\dom}{{\rm dom} \, }

\newcommand{\tcb}{\textcolor{black}}

\newcommand{\norm}[1]{\left\lVert#1\right\rVert}

\newtheorem{assume}{Assumption}
\newtheorem{thm}{Theorem}
\newtheorem{prop}{Proposition}
\newtheorem{cor}{Corollary}
\newtheorem{lemm}{Lemma}
\newtheorem*{pr2}{Proof of Theorem 2}

\title{\LARGE \bf
\tcb{Analyzing the Effect of Persistent Asset Switches on a Class of Hybrid-Inspired Optimization Algorithms}}

\author{Matina Baradaran, Justin H. Le, Andrew R. Teel
\thanks{This work is supported by AFOSR grant FA9550-18-1-0246. Matina Baradaran, Justin H. Le, and Andrew R. Teel are with the Department of Electrical and Computer Engineering, University of California, Santa Barbara, CA 93106, USA. Email: 
        {\tt\small \{baradaranhosseini, justinle, teel\}@ucsb.edu}}
}

\begin{document}
\maketitle
\thispagestyle{empty}
\pagestyle{empty}
\begin{abstract}
Convex optimization challenges are currently pervasive in many science and engineering domains. In many applications of convex optimization, such as those involving multi-agent systems and resource allocation, the objective function can persistently switch during the execution of an optimization algorithm. Motivated by such applications, we analyze the effect of persistently switching objectives in continuous-time optimization algorithms. In particular, we take advantage of existing robust stability results for switched systems with distinct equilibria and extend these results to systems described by differential inclusions, making the results applicable to recent optimization algorithms that employ differential inclusions for improving efficiency and/or robustness. Within the framework of hybrid systems theory, we provide an accurate characterization, in terms of Omega-limit sets, of the set to which the optimization dynamics converge. Finally, by considering the switching signal to be constrained in its average dwell time, we establish semi-global practical asymptotic stability of these sets with respect to the dwell-time parameter.
\end{abstract}

\section{INTRODUCTION}

Convex optimization challenges are pervasive across many current science and technology fields. Such optimization problems are often solved using iterative algorithms such as first-order gradient-based methods that can be naturally represented and analyzed as dynamical systems. However, most studies of these algorithms do not account for applications in which the objective function to be optimized can instantaneously \tcb{change} at discrete moments in time during the algorithm's execution. Switching objectives arise in increasingly many real-world applications, such as multi-agent systems in which the agents must be replaced according to a real-time mission constraint, as well as resource allocation problems in which the allocated assets can experience persistent changes in the reward they generate. Other examples can be found in branches of human science, such as sociology, psychology, and organization science, which study the group interactions and performance of teams with multiple individuals engaged in a common task \cite{anderson2014dynamic}, \cite{march1991exploration}. In \cite{appraisalnetwork}, for instance, an algorithm is proposed and studied for optimizing the sum of the team members' performance measures, and each member's performance measure is determined by variables such as the member's specific skill level. Thus, the objective will switch whenever a team member is replaced during execution of the algorithm. 
Persistent switches of assets could also play an important role in the realm of intelligent control of unmanned aerial vehicles (UAVs), also referred to as drones. In recent years, there has been a surge in the use of UAVs for surveillance and security, parcel shipment, traffic monitoring, disaster recovery, and military reconnaissance \cite{menouar2017uav},\cite{zhou2015multi}. Cooperative and collaborative optimization of UAV performance is essential in such applications. In \cite{bekmezci2013flying}, applications and engineering constraints for UAV-based mobile ad hoc networking are surveyed. In all these settings, there are few studies of how optimization dynamics are impacted by objectives that persistently switch due to failures or replacements of UAVs in a network/relay. In this paper, we analyze how the presence of a persistently switching objective impacts the asymptotic stability properties of optimization dynamics. Our analysis aims to be applicable to systems such as the aforementioned UAV ad hoc networks, where the UAVs that suffer from low battery, potential damages, or other disabling aspects may be replaced during execution of optimization dynamics. Our analysis targets applications where a system must instantaneously replace assets (UAVs, team members, etc.) at discrete moments in time, while it optimizes performance continuously in time. To prevent instability and be able to characterize the set to which the optimization dynamics converge, such switches should satisfy an average dwell-time constraint. In \cite{baradaranIFAC2020}, stability is studied for systems involving switches between multiple differential equations with distinct equilibria, with switches satisfying an average dwell-time condition. We extend the results from \cite{baradaranIFAC2020} and apply the provided asymptotic characterization method to consider switching between differential inclusions with distinct equilibria. The differential inclusions each take the form of the Hybrid-inspired Heavy Ball System from \cite{JustinLe2020}, which takes advantage of the differential inclusion to achieve efficiency comparable to accelerated gradient methods while also retaining certain robustness properties. We characterize the set to which the resulting switched system converges, in terms of the Omega-limit set of an associated ideal hybrid system. This ideal hybrid system involves an automaton with solutions that are in one-to-one correspondence with time domains satisfying the average dwell-time constraint and with the rate parameter set to zero. Finally, we show that the system switching between differential inclusions with a small disturbance is a perturbed version of the mentioned ideal hybrid system with a globally asymptotically stable Omega-limit set. Thus, we can establish semi-global, practical asymptotic stability for the perturbed system. \tcb{The robust stability of switched systems with multiple equilibria is a well studied subject (see \cite{alpcan2010stability},\cite{veer2017generation}). For example, \cite{GianlucaJorge2020} studies the steady-state optimization of switched systems with time-varying cost functions.}
We demonstrate our stability results on an application involving a data relay formed by UAVs, where the  objective models various performance measures such as the communication quality or battery consumption. \vspace{-3pt}
\section{Notations}
In this work, $\Re^n$ is used to demonstrate the $n$-dimensional Euclidean space. The $\Re_{\geq 0}$ is used to show the nonnegative real numbers. We use $|x|$ to denote the Euclidean norm of the verctor $x\in \Re^n$.  Given $r>0$, we use $r \Ball$ for the set $\{x \in \Re^n: |x|  \leq r \}$. For a function $\alpha$, we say $\alpha \in \mathcal{K}^{+}$ if $\alpha: \Re_{\geq 0} \rightarrow \Re_{\geq 0}$ is continuous and strictly increasing. 
\section{Hybrid Systems}
The hybrid systems framework that we use is described in \cite{Goebel12a}. Formally, hybrid systems are modeled as
 \begin{subequations}
 \label{eq:1}
 \begin{align}
     x \in C , & \qquad \ \dot{x} \in F(x) \\
     x \in D , & \qquad x^{+} \in G(x),
 \end{align}
 \end{subequations}
 where $x \in \Re^{n}$ is the state, $C \subset \Re^{n}$ is the flow set, $D \subset \Re^{n}$ is the jump set,
 $F:\Re^{n} \rightrightarrows \Re^{n}$ is the flow map, and $G:\Re^{n} \rightrightarrows \Re^{n}$ is the jump map.
 The data $(C,F,D,G)$ is said to satisfy the {\em hybrid basic conditions} if $C$ and $D$ are closed, the graphs of $F$ and $G$ are closed, $F$ and $G$ are locally bounded, the values of $F$ are nonempty and convex on $C$ and the values of $G$ are nonempty on $D$. We use $\mathcal{S}(K)$ to denote the {\em set of solutions} to (\ref{eq:1}) that start in $K$. 
\section{Preliminaries}
In the first part of this section, we introduce some notions later used in this work and state some stability concepts for hybrid systems. In the second part of this section, we review the concept of average dwell-time constraint \cite{HespanhaMorse99CDC}. Next, we consider a perturbed hybrid system, in which the perturbation is parameterized by a dwell-time parameter satisfying the average dwell-time condition from \cite{Cai2008smooth}. Finally, we discuss the advantages of optimization methods employing differential inclusions for applications with persistent asset switches, introducing an efficient differential inclusion-based optimization algorithm inspired by the hybrid algorithms studied in \cite{JustinLe2020}. This optimization algorithm has been shown in \cite{JustinLe2020} to have desirable robustness properties and will be shown in Section \ref{sec:numerical} to be especially suitable for online optimization problems with persistent asset switches. \vspace{-3pt}
\subsection{Stability concept for hybrid systems} 
We state some stability analysis concepts particularly for the hybrid system framework (\ref{eq:1}). The hybrid system (\ref{eq:1}) is said to be Lagrange stable if there exists $\alpha \in \mathcal{K}^{+}$ such that, for each $z_0 \in \Re^n$, each $x \in \mathcal{S}(x_0)$ and $(t,j) \in \dom(x)$, we have $|x(t,j)|\leq \alpha(|x_0|)$. A compact set $\mathcal{A} \subset \Re^{n}$ is said to be {\em stable} for the hybrid system (\ref{eq:1}) if,
for each $\varepsilon>0$, there exists $\delta>0$ such that $|x_{\circ}|_{\mathcal{A}} \leq \delta$, $x \in \mathcal{S}(x_{\circ})$ and $(t,j) \in \mbox{\rm dom}(x)$ imply that $|x(t,j)|_{\mathcal{A}} \leq \varepsilon$. A compact set $\mathcal{A} \subset \Re^{n}$ is said to be {\em attractive} for (\ref{eq:1}) if there exists $\delta >0$ such that each solution
$x \in \mathcal{S}(\mathcal{A}+\delta \Ball)$ is bounded and, if complete, satisfies $\lim_{t+j \rightarrow \infty} |x(t, j)|_{\mathcal{A}}=0$. The {\em basin of attraction} for an attractive set $\mathcal{A}$ is the set of initial conditions from which each solution is bounded and, if complete, satisfies $\lim_{t+j \rightarrow \infty} |x(t, j)|_{\mathcal{A}}=0$. A compact set $\mathcal{A} \subset \Re^{n}$ is said to be {\em asymptotically stable} for (\ref{eq:1}) if it is stable and attractive. It is said to be {\em globally asymptotically stable} for (\ref{eq:1}) if it is
asymptotically stable with $\Re^n$ as its basin of attraction. The set $\mathcal{A} \subset \Re^{n}$ is said to be {\em semiglobally practically asymptotically stable in the parameter $\delta>0$} for the perturbed hybrid system if there exists $\beta \in \mathcal{K}\mathcal{L}$ and, for each
$\varepsilon>0$ and $\Delta>0$, there exists $\delta>0$ such that each
$x \in \mathcal{S}_{\delta}(\mathcal{A} + \Delta \Ball)$ satisfies
$|x(t,j)|_{\mathcal{A}} \leq \beta(|x(0,0)|_{\mathcal{A}},t+j)+\tcb{\varepsilon}$ for all $(t,j) \in \mbox{\rm dom}(x)$.
\subsection{Average dwell time switching and its automaton} 
\label{prep}
Let a family of differential inclusions be given
\begin{align}
\label{fam}
    \dot{z} \in \overline{\mbox{co}} F_{\sigma}(z+ \delta \Ball) + \delta \Ball, \qquad \sigma \in \Sigma. \vspace{-20pt}
\end{align}
Let $\Sigma:=\{1,\cdots,M\}$ with $M$ being a positive integer. The switching signal is denoted by $\sigma: \Re_{\geq 0} \to \Sigma$ which satisfies the average dwell-time constraint parameterized by a small $\delta >0$. We formalize the concept of average dwell time for the switching signal $\sigma$. Let $N_{\sigma}[s,t]$ denote the number of switches $\sigma$ within the interval $[s,t]$ and $N_0$ being a positive integer. We then assume \vspace{-5pt}
\begin{align}
\label{eq:dwell}
N_{\sigma}(s,t) \leq \delta (t-s) + N_{0}   \qquad \forall \ 0 \leq s \leq t . 
\end{align} 
We cast a hybrid system, which associates the dynamics from differential inclusions given in (\ref{fam}) and hybrid time domains satisfying the constraint from (\ref{eq:dwell}). This hybrid system employs an automaton to capture the average dwell-time condition. The model of this hybrid system $\hat{\mathcal{H}}_\delta$ is given 
\vspace{-15pt} \begin{align}
\label{eq:Hdelta} 
\hat{\mathcal{H}}_{\delta} 
\begin{cases}
                             (z,\sigma, \tau) \in C \! \times  \! \Sigma  \! \times \!  [0,N_{0}]  \
                             \begin{cases}
                                   \dot{z} & \! \! \! \! \in  \overline{\mbox{co}} F_{\sigma}(z \! + \! \delta \Ball) \! + \! \delta \Ball \\
                                   \dot{\sigma}  &  \! \! \! \! = 0 \\
                                   \dot{\tau} & \! \! \! \! \in [0,\delta]
                             \end{cases} \\
                              (z,\sigma,\tau) \in C \! \times \! \Sigma \! \times \! [1,N_{0}] \
                              \begin{cases}
                              z^{+} & \! \! \! \! = z \\
                              \sigma^{+} & \! \! \! \! \in \Sigma \backslash \left\{ \sigma \right\} \\
                             \tau^{+}  & \! \! \! \! = \tau - 1 .
                             \end{cases}
\end{cases} 
\end{align}
 According to \cite[Proposition 1.1]{Cai2008smooth}, the solutions to $\hat{\mathcal{H}}_{\delta}$ are in a one-to-one correspondence with the solutions of (\ref{fam}) under the average dwell-time switching constraint (\ref{eq:dwell}) and $\hat{\mathcal{H}}_{\delta}$ satisfies the hybrid basic conditions from \cite[Assumption 6.5]{Goebel12a}. Note that, for convenience, the parameter $\delta >0$ describes both the maximum flow rate of the automaton $\tau$ and the perturbed differential inclusion. \vspace{-8pt}
 \subsection{Ideal system analysis}
 The goal is to extend the results from \cite{baradaranIFAC2020} to differential inclusions and characterize the asymptotic behavior of (\ref{fam}) under average dwell-time constraint parametrized by a small $\delta$. In order to realize this goal, we characterize the asymptotic behavior of the $\hat{\mathcal{H}}_\delta$ without the perturbation, which results from setting $\delta = 0$ in (\ref{eq:Hdelta}). This new ideal hybrid system obtained by setting $\delta = 0$ is denoted by $\hat{\mathcal{H}}_{0}$ and corresponds to
\begin{align}
\label{eq:H0} 
\hat{\mathcal{H}}_{0} 
\begin{cases}
                             (z,\sigma, \tau) \in C \! \times  \! \Sigma  \! \times \!  [0,N_{0}]  \
                             \begin{cases}
                                   \dot{z} & \! \! \! \! \in F_{\sigma}(z)  \\
                                   \dot{\sigma}  & \! \! \! \! = 0 \\
                                   \dot{\tau} & \! \! \! \! = 0
                             \end{cases} \\
                              (z,\sigma,\tau) \in C \! \times \! \Sigma \! \times \! [1,N_{0}] \
                              \begin{cases}
                              z^{+} & \! \! \! \!  = z \\
                              \sigma^{+} & \! \! \! \! \in \Sigma \backslash \left\{ \sigma \right\} \\
                             \tau^{+}  &  \! \! \! \! =  \tau - 1 .
                             \end{cases}
                           \end{cases}
\end{align}
As we will see, the asymptotic behavior of (\ref{eq:H0}) approximates the asymptotic behavior of the (\ref{eq:Hdelta}). 
\subsection{A differential inclusion-based optimization algorithm}
\label{meth}
Differential inclusions have been useful for studying stability of steepest descent/ascent dynamics in convex optimization \cite{goebel2019glimpse}. Differential inclusions are also discussed in \cite{franca2018nonsmooth}, in which a continuous-time analogue of the Alternating Direction Method-of-Multipliers is given. In \cite{Adri20201} and \cite{Adri20202}, differential inclusions are used to approximate a high-gain anti-windup strategy for handling input constraints in feedback-based optimization. In this section, we focus on an algorithm, represented as a differential inclusion, whose trajectories seek a solution to the \tcb{problem} \vspace{-5pt}
\begin{subequations}
\label{eq:problem_unswitched}
\begin{align}
    &\min_{q \in \Re^n} \left\{\phi(q) \coloneqq \sum_{i=1}^n \label{1a} 
      \phi_{i}(q_i)\right\} \\
    &\text{s.t.} \quad \mathds{1}^{T} q = d, \quad d \in \Re, \label{eq:balance_constraint}
\end{align}
\end{subequations}
under the following assumptions. 
\begin{assume}
\label{objective_assumptions}
\tcb{The objective} $\phi: \Re^n \to \Re$ is continuously differentiable, has compact sub-level sets, has an $L$-Lipschitz gradient $\nabla \phi$, and is convex.
\end{assume}
Under Assumption \ref{objective_assumptions}, \cite[Prop. 5.3.7]{bertsekas-2009} implies that the set $\mathcal{Q}^*$ of solutions to \eqref{eq:problem_unswitched} is non-empty and compact. Furthermore, \cite[Prop. 5.3.3]{bertsekas-2009} implies that $q^* \in \mathcal{Q}^*$ if and only if there exists $\mu^* \in \Re$ such that
\begin{align}
\label{eq:kkt_original}
    \nabla\phi(q^*) + \mu^* \mathds{1} = 0, \quad \mathds{1}^T q^* = d.
\end{align}
Defining $\mathcal{L}$ as the Laplacian of a connected undirected graph, we have that the nullspace of $\mathcal{L}$ is $\text{span}(\mathds{1})$ \cite[Sec. II]{olfati-saber-2007}, and thus, the conditions \eqref{eq:kkt_original} can be equivalently expressed as
\begin{align}
\label{eq:kkt}
    \mathcal{L}\nabla\phi(q^*) = 0, \quad \mathds{1}^T q^* = d.
\end{align}
\tcb{It is convenient} to express the above conditions in terms of $\mathcal{L}$ because $\mathcal{L}$ will play an important role in the analysis and application of our proposed algorithm. Our algorithm is based on the following differential equation, which we refer to as the \textit{Laplacian-Gradient Heavy Ball Method} (HBM) and is a continuous-time analogue of the algorithm in \cite{ghadimi2013multi}. The HBM is given by \vspace{-8pt} 
\begin{align}
    \dot x &= \left[\begin{array}{cc}
        p \\
        -Kp - \mathcal{L}\nabla\phi(q)
    \end{array}\right].  \label{eq:hbm}
\end{align}
To achieve fast convergence without oscillations, first-order convex optimizations methods often require knowledge of problem parameters for a precise algorithmic parameter tuning. Theoretically, convergence of iterative optimization algorithms for convex problems can be achieved through momentum in the sense of Nesterov's method \cite{Nesterov_method}. However, often the algorithmic parameters such as momentum and stepsize should be accurately specified depending on the problem parameters. When such precise parameters are not known, momentum methods such as Nesterov’s method suffer from oscillations in their trajectories. Such oscillations exhibited by Nesterov's method often restrict its application to real-world systems \cite{hauswirth-2020}. 
Therefore, we consider the differential inclusion, based on HBM and inspired by the hybrid algorithm in \cite{Teel2019}, which can greatly reduce oscillations without precise algorithmic parameter tuning. This algorithm is referred to as the \textit{Laplacian-Gradient Hybrid-inspired Heavy Ball Method} (HiHBM). This approach using a differential inclusion can also improve transient performance in settings where the objective is persistently switching during the algorithm execution, as we show in Sec.~\ref{sec:numerical}. In this system, the state is denoted as $x \coloneqq (q, p)$. The parameters $\{\underline{K}, \overline{K}\} \in \Re^2$ satisfy $0 < \underline{K} \leq \overline{K}$. The system is defined as \vspace{-11pt}
\begin{subequations}
\label{eq:hihbm}
\begin{align}
    \dot x &\in F(x) \coloneqq \left[\begin{array}{cc}
        p \\
        -\kappa(x)p - \mathcal{L}\nabla\phi(q)
    \end{array}\right], \label{eq:hihbm_flow_map} \\
    \kappa(x) &\coloneqq \kappa(x; \underline{K}, \overline{K}) \nonumber \\
    &\coloneqq \begin{dcases}
        \overline{K} & \text{if } \langle\mathcal{L}\nabla\phi(q), p\rangle > 0, \\
        \underline{K} & \text{if } \langle\mathcal{L}\nabla\phi(q), p\rangle < 0, \\
        {\left[\underline{K}, \overline{K}\right]} & \text{if } \langle\mathcal{L}\nabla\phi(q), p\rangle = 0,
    \end{dcases} \label{eq:switch} \\
    C &\coloneqq \{(q, p) \in \Re^{2n}: \; \mathds{1}^Tq = d \text{ \& } \mathds{1}^Tp = 0\}.  \label{eq:Cflowmap}
\end{align}
\end{subequations} 
To establish the following global asymptotic stability property of HiHBM, it will be convenient to write $\mathcal{F}_d$ to denote the feasible set defined by \eqref{eq:balance_constraint}. That is, we have
\begin{align}
    \mathcal{F}_d &\coloneqq \{q \in \mathbb R^n: \mathds{1}^Tq = d\}, \label{eq:feasible_set}
\end{align}
and the set defined in \eqref{eq:Cflowmap} can be written $C = \mathcal{F}_d \times \mathcal{F}_0$.

\begin{thm}
\label{thm:gas_hihbm}
Under Assumption~\eqref{objective_assumptions}, the set $\mathcal{A} \coloneqq \mathcal{Q}^* \times \{0\}$ is GAS for the system \eqref{eq:hihbm}.
\end{thm}
\begin{proof}
Let $\mathcal{L}^{\dagger}$ denote the generalized inverse Laplacian \cite{Gutman04}, which can be shown to be a symmetric positive semi-definite matrix of rank $n-1$ satisfying \vspace{-5pt}
\begin{align}
    \label{eq:inverse_laplacian}
    \mathcal{L}^{\dagger}\mathds{1} = 0, \qquad \mathcal{L}\mathcal{L}^{\dagger} = I_n - \frac{1}{n}\mathds{1}\mathds{1}^T.
\end{align}
We aim to apply \cite[Thm. 8.2]{Goebel12a}. Toward this goal, consider
\begin{align}
    V(q, p) &\coloneqq \phi\left(q\right) - \phi^* + \frac{1}{2}p^T\mathcal{L}^{\dagger}p, \\
    \phi^* &\coloneqq \min_{w \in \mathcal{F}_d} \phi(w). \label{eq:phi^*}
\end{align}
The map $q \mapsto \phi(q) - \phi^*$ is positive definite on $C$ with respect to $\mathcal{Q}^*$ because $q \in \mathcal{F}_d$ for all points in $C$. The map $p \mapsto p^T\mathcal{L}^{\dagger}p$ is positive definite on $C$ with respect to $\{0\}$, which follows from \eqref{eq:inverse_laplacian} and the fact that $p \in \mathcal{F}_0$ for all points in $C$. By Assumption~\eqref{objective_assumptions}, $V$ is radially unbounded with respect to $\mathcal{A}$ relative to $C$. The Lie derivative of $V$ with respect to \eqref{eq:hihbm_flow_map} satisfies \vspace{-7pt}
\begin{align*}
    &\langle\nabla\phi(q), \; p\rangle - \left\langle\mathcal{L}\nabla\phi(q), \; \mathcal{L}^{\dagger} p\right\rangle - \kappa(x)p^T \mathcal{L}^{\dagger} p \nonumber \\
    &= p^T\left(\frac{1}{n}\mathds{1}\mathds{1}^T\right) \nabla\phi(q) - \kappa(x)p^T \mathcal{L}^{\dagger} p \nonumber = - \kappa(x)p^T \mathcal{L}^{\dagger} p \leq 0, \vspace{-15pt}
\end{align*} 
where the first equality follows from \eqref{eq:inverse_laplacian}, and the second equality follows from the fact that $\mathds{1}^Tp = 0$. Stability of $\mathcal{A}$ has now been shown. To show attractivity, first recall that $p \mapsto p^T\mathcal{L}^{\dagger}p$ is positive definite when restricted to $\mathcal{F}_0$, and thus, $\kappa(x)p^T \mathcal{L}^{\dagger} p$ can remain at zero for all time only if $p$ remains at zero. However, if $p$ remains at zero, then \eqref{eq:hihbm_flow_map} implies that both $\dot{q}$ and $\mathcal{L}\nabla\phi(q)$ must remain at zero, which can happen only if $q$ remains in $\mathcal{Q}^*$, due to \eqref{eq:kkt}. In summary, $\kappa(x)p^T \mathcal{L}^{\dagger} p$ remains at zero only if $(q, p)$ remains in $\mathcal{A}$, and the result then follows from \cite[Thm. 2.11]{Ryan98}.
\end{proof} 
\section{Online optimization with persistent switches}
\label{sec:drones} 
In this section, we discuss an optimization problem that models a team of drones collectively executing a task. The drones are the assets that persistently switch in this scenario, due to each drone's limited battery life, potential physical damages, or other disabling aspects. Consider a team of $n$ drones with the collective task of forming a relay to transmit data from a source to a destination. The drones form a straight path of length $d \in \Re_{>0}$ from the first to $n$-th drone. The network-wide state vector is denoted by $q \in \Re^n$, where $q_i$ is the relative distance from drone $i$ to the node that precedes it, while the data source is considered to precede the first drone. We assume that the drones $1$ to $n$ never cross over their neighboring drones, and their ordering remains the same during their movement. 
In practice, enforcement of a box constraint on each $q_i$ is needed, which can be done using the penalty-based approach in \cite{cherukuri-2015}, but the details are beyond the scope of this work. The problem can be modeled as a variant of the problem \eqref{eq:problem_unswitched}, in which the objective switches according to a switching signal $\sigma: \Re_{\geq 0} \to \Sigma$:
\begin{subequations}
\label{optss}
\begin{align}
    &\min_{q \in \Re^n} \left\{\phi_{\sigma(t)}(q) \coloneqq \sum_{i=1}^n
    \phi_{{i},{\sigma(t)}}(q_i)\right\} \\
    &\text{s.t.} \quad \mathds{1}^{T} q = d.
\end{align}
\end{subequations}
In a similar setting to (\ref{fam}), the switching signal satisfies the average dwell-time constraint parametrized by a small $\delta$.
\begin{assume}
\label{new_objective_assumptions}
For each constant $\sigma(.) \in \Sigma$, the function $q \mapsto \phi_{\sigma}(q)$ satisfies the conditions of Assumption \ref{objective_assumptions} and there exists a unique solution $q^*$ to \eqref{optss}.
\end{assume}
\begin{thm}
\label{thm1}
\tcb{
If Assumptions \ref{objective_assumptions} and \ref{new_objective_assumptions} hold then the $\Omega$-limit set associated with the ideal hybrid system
$\hat{\mathcal{H}}_0$ given in (\ref{eq:H0}) with $z = (q, p)$ and with \vspace{-9pt}
\begin{align}
 F_{\sigma}(z) &\coloneqq \left[\begin{array}{cc}
        p \\
        -\kappa(z)p - \mathcal{L}\nabla\phi_{\sigma}(q)
    \end{array}\right], \nonumber \\
    &\coloneqq \begin{dcases}
        \overline{K} & \text{if } \langle\mathcal{L}\nabla\phi_{\sigma}(q), p\rangle > 0, \\
        \underline{K} & \text{if } \langle\mathcal{L}\nabla\phi_{\sigma}(q), p\rangle < 0, \\
        {\left[\underline{K}, \overline{K}\right]} & \text{if } \langle\mathcal{L}\nabla\phi_{\sigma}(q), p\rangle = 0,
    \end{dcases}  \\
    C &\coloneqq \{(q, p) \in \Re^{2n}: \; \mathds{1}^Tq = d \text{ \& } \mathds{1}^Tp = 0\}, 
\end{align}
is semiglobally, practically asymptotically stable in the parameter $\delta >0$ for the perturbed system $\hat{\mathcal{H}}_\delta$ given in (\ref{eq:Hdelta}).}
   \begin{figure}
    \centering
    \includegraphics[scale=0.32]{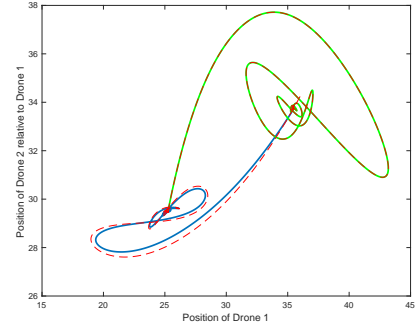}
    \caption{Steady State Behavior near Omega-limit set}
    \label{fig:omega}
       \vspace{-10pt}
\end{figure}
\begin{figure}
     \centering
     \begin{subfigure}[b]{0.35\textwidth}
         \centering
         \includegraphics[width=\linewidth]{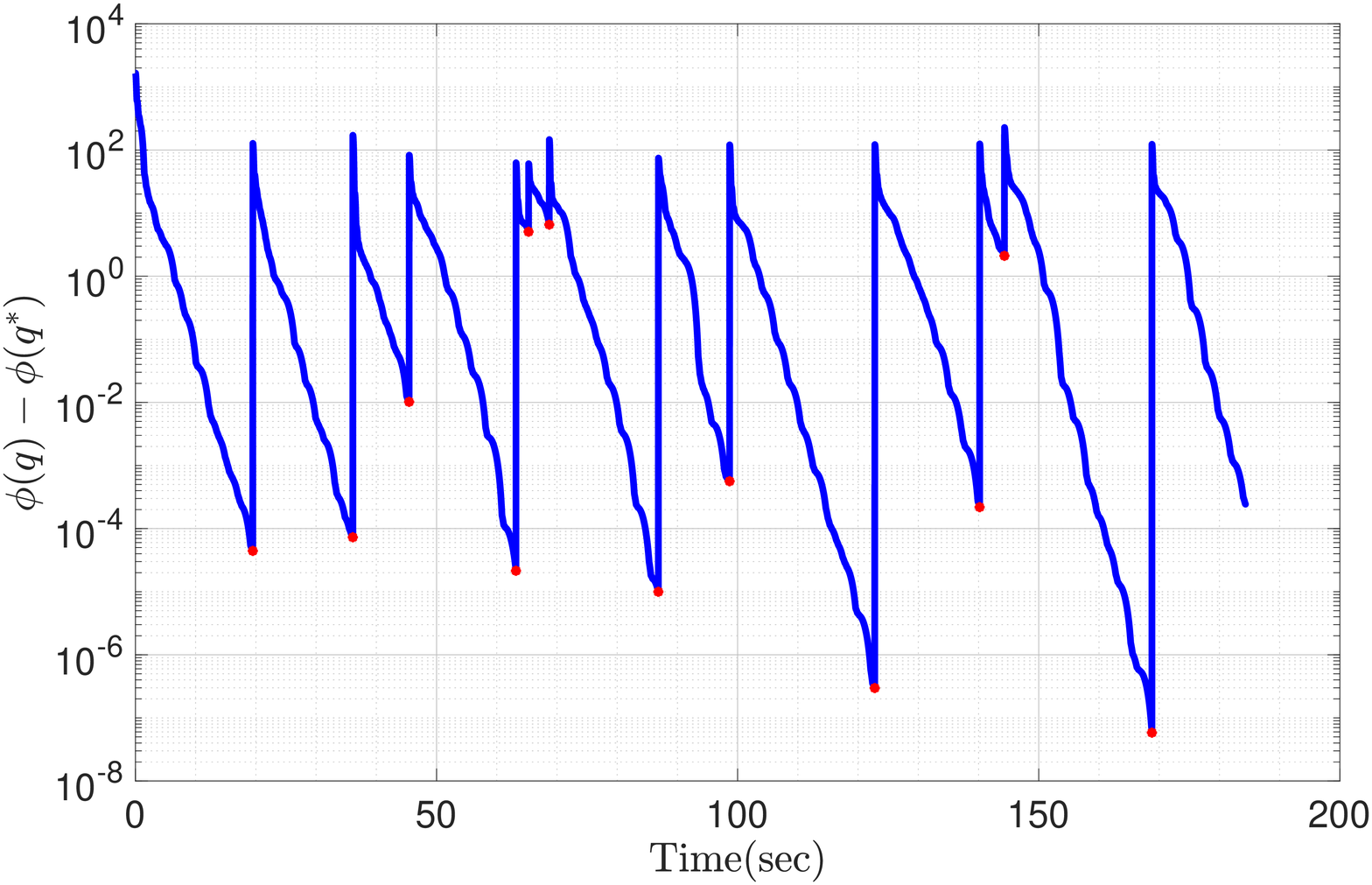}
         \caption{Optimization with persistent switches}         
         \label{fig:phi}
     \end{subfigure}
     \hfill
     \begin{subfigure}[b]{0.35\textwidth}
         \centering
         \includegraphics[width=\linewidth]{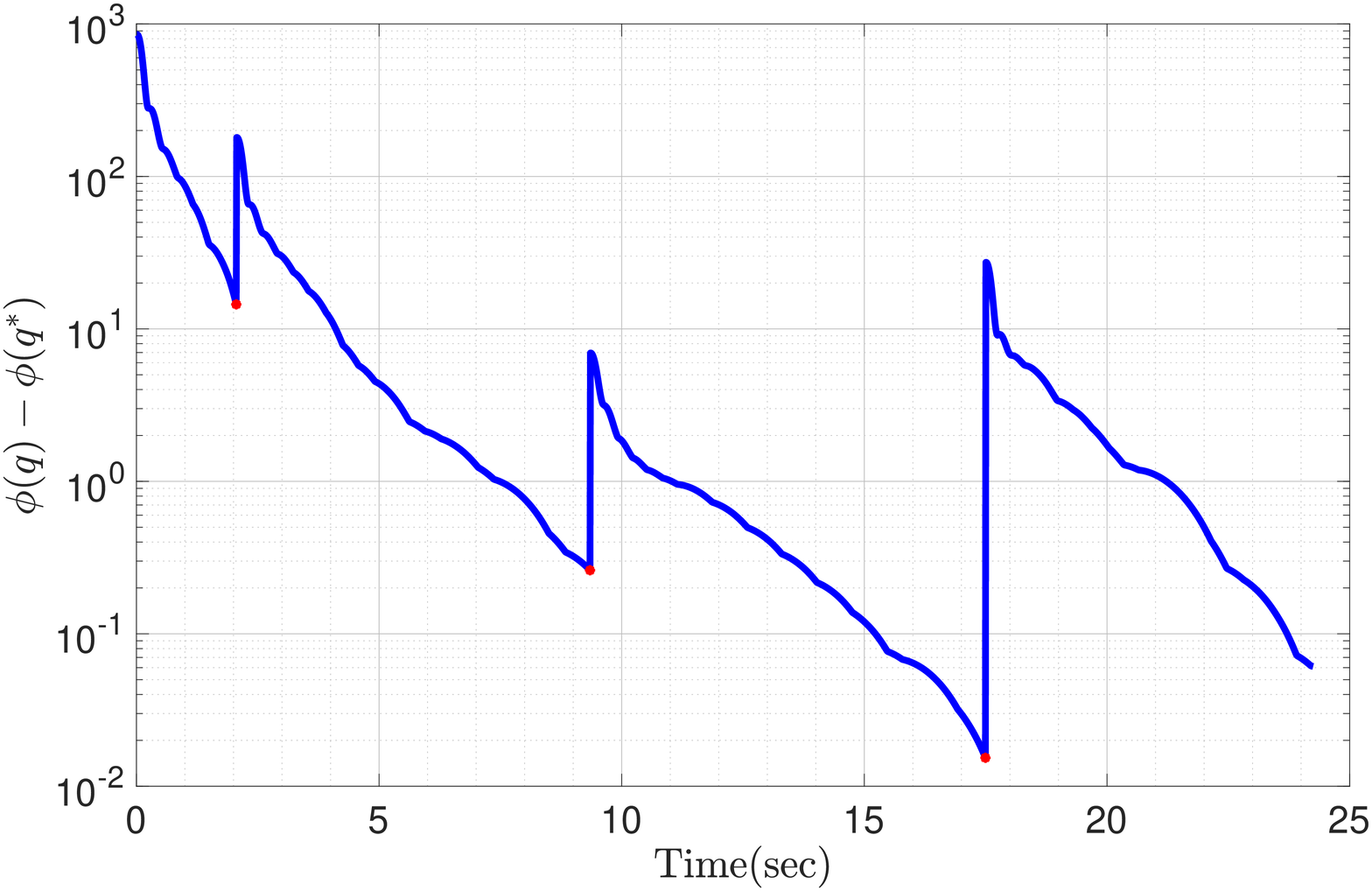}
         \caption{Optimization with fewer asset switches}                 
         \label{fig:phi2b}
     \end{subfigure}     
        \caption{}
        \label{fig:20drones}
               \vspace{-18pt}
\end{figure}
\end{thm}
  \begin{figure}
    \centering
    \includegraphics[scale=0.20]{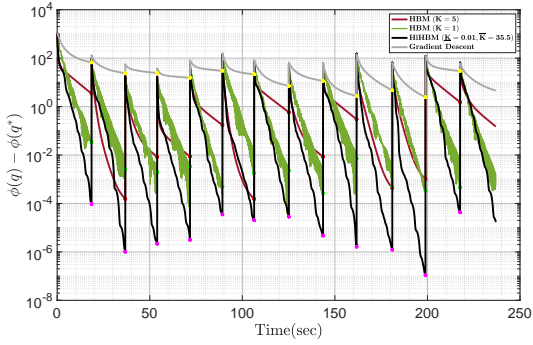}
    \caption{Comparison of HBM (\ref{eq:hbm}) and HiHBM (\ref{eq:hihbm})}
    \label{fig3}
       \vspace{-18pt}
\end{figure}
\vspace{-5pt} The proof of Theorem \ref{thm1} is given in Section \ref{proofsec}. \vspace{-5pt}
\section{Numerical results}\label{sec:numerical} 
In this section, we give numerical examples for the application and the \tcb{optimization} method discussed in Sections \ref{sec:drones} and \ref{meth} respectively. In Figure \ref{fig:omega}, we consider a team of $2$ drones in a relay of distance $d=100$, and we consider $N_0 = 1$ in order to demonstrate the $\Omega$-limit set. The blue line represents the position of the first drone (initial condition: distance of $35.5071$ units from the data source), and the green line represents the position of a second drone relative to the first drone in the relay (initial condition: distance of $33.7398$ units from the first drone). These two lines depict the $\Omega$-limit set of the ideal hybrid system given in (\ref{eq:H0}). The red lines are created by allowing for the small dwell-time parameter $\delta = 0.0338$. The switches satisfy the average dwell-time condition. Considering the switching behavior to be a small disturbance as in Theorem \ref{thm1}, we observe that the solutions converge to a small neighborhood of the $\Omega$-limit set of the ideal hybrid system. \tcb{These observations agree with the results from Theorem \ref{thm1}}. Figure \ref{fig:phi} shows the decrease in the value of the objective function under the dynamics of the HiHBM algorithm applied to the problem (\ref{optss}) for $20$ drones in a relay of length $d= 100$, with the network-wide objective being $\phi_{\sigma}(q) = \frac{1}{2} q^T P_{\sigma} q + b_{\sigma}^Tq$ for $\sigma \in \{1, 2\}$. Each $P_{\sigma}$ is diagonal with eigenvalues i.i.d. uniform on $[10, 20]$, and each $b_{\sigma}$ has its entries i.i.d. uniform on $[-10, 10]$. The objective at each drone is $\phi_{i,\sigma}(q_{i}) = \frac{1}{2}P_{\sigma, ii}q_i^2 + b_{\sigma, i}q_i$, where $q_i$ is the distance from drone $i$ to drone $i - 1$, as described in Sec.~\ref{sec:drones}. We allow persistent switches satisfying average dwell time with $\delta = 0.06$. Figures \ref{fig:phi} shows that HiHBM converges efficiently even as the objective switches persistently. Figure \ref{fig:phi2b} shows the same example with fewer switches. In Figure \ref{fig3}, the gray line displays the gradient descent optimization. The red line demonstrates the convergence of the HBM optimization algorithm with $K=5$ and the green line demonstrates the same optimization algorithm with $K=1$. The black line in Figure \ref{fig3} shows the HiHBM optimization algorithm with a faster, more efficient convergence. The lower and upper-bound for HiHBM are $0.01$ and $35.5$ respectively. In Figure \ref{fig3}, we see the efficiency of the HiHBM algorithm that generates smaller errors in comparison with the simple gradient descent and the HBM method. 
\section{The proof of theorem \ref{thm1}}
\label{proofsec}
\subsection{A new result on switching between constrained inclusions}
\label{sec:spas} Motivated by applications requiring persistent asset switches while optimizing, in this subsection we provide results on characterization of the asymptotic behavior that results from persistent switches among asymptotically stable differential inclusions with distinct equilibria. The switching signal is considered to satisfy an average dwell-time constraint as mentioned in the preliminaries \ref{prep}.
\subsubsection{Assumptions} We have two assumptions regarding the following family of differential inclusions parametrized by $\sigma \in \Sigma$ given as
\vspace{-12pt}
\begin{align}
\label{eq:inclu}
\dot{z} \in F_{\sigma}(z) \quad z\in C.
\vspace{-10pt}
\end{align}
\begin{assume}
\label{assume:2}
For each $\sigma \in \Sigma$, $F_{\sigma}$ is outer semi-continuous, locally bounded relative to $C \subset \dom F_q$, and, for each $z$, $F_{\sigma}(z)$ is non-empty and convex for all values of $z \in C$. Furthermore, the point $z_{\sigma}^*$ is globally asymptotically stable for (\ref{eq:inclu}).
\end{assume}
\begin{assume}
\label{assume:3}
For each $\sigma \in \Sigma$, the only solution to 
\begin{align}
\label{sup}
\bigg\{ z \in C,  \quad \dot{z} \in -F_{\sigma}(z)
\end{align} \vspace{-4pt}
with an initial value of $z(0) = z^*_{\sigma}$ is $z(t) = z^*_{\sigma}$ for all $t \geq 0$.
\end{assume}
Assumption \ref{assume:3} extends \cite[Assumption 2]{baradaranIFAC2020} from differential equations to differential inclusions. \vspace{-4pt}
\subsection{Extended Main Result}
In this section, we adapt the results from \cite{{baradaranIFAC2020}} to the hybrid system $\hat{\mathcal{H}}_{\delta}$. 
First, we establish boundedness for solutions to $\hat{\mathcal{H}}_{0}$ by claiming the following proposition and corollary. 
\begin{prop}
\label{prop1}
If, for each $\sigma \in \Sigma$, the system (\ref{eq:inclu}) is Lagrange stable, the hybrid system $\hat{\mathcal{H}}_{0}$ is Lagrange stable. 
\end{prop}
The proof of Proposition \ref{prop1} follows the boundedness results from \cite[Section 4.2]{baradaranIFAC2020} for the ideal hybrid system $\hat{\mathcal{H}}_{0}$.
As a consequence of Proposition \ref{prop1} we have the following corollary. 
\begin{cor}
\label{cor1}
If Assumption \ref{assume:2} holds then the hybrid system $\hat{\mathcal{H}}_{0}$ is Lagrange stable. 
\end{cor}
Using the definition of $\Omega$-limit set from \cite{baradaranIFAC2020}, the rest of this section is devoted to the characterization of the $\Omega$-limit set of $\hat{\mathcal{H}}_{0}$ from a compact set $K \subset \Re^{n+2}$ denoted by $\Omega_0(K)$. We use the definition of reachable sets from \cite{baradaranIFAC2020} and demonstrate the reachable set from $K$ with $R_{0}(K)$. Following the setting from \cite{baradaranIFAC2020}, we define
\begin{subequations}
\label{eq:setSS}
\begin{align}
S_{\sigma} & := \bigcap_{j \in \ZZ_{\geq 0}} \overline{ R_{0}\left( \vphantom{\frac{1}{2}} \left( \vphantom{\frac{}{}} \left\{ z_{\sigma}^{*} \right\} + \frac{1}{j+1} \Ball \right) \times \left\{\sigma \right\} \times [0,N_{0}] \right)} \\
S  & : = \bigcup_{\sigma \in \Sigma} S_{\sigma} \label{eq:setS} .  
\end{align}
\end{subequations}
\begin{lemm}
\label{lem0}
Under Assumption \ref{assume:2} and as a result of Corollary \ref{cor1}, the set $S$ is compact. 
\end{lemm}
\begin{prop}
\label{prop11}
Under Assumptions \ref{assume:2} and \ref{assume:3},
for each compact set $K$ with $\left( \bigcup_{\sigma \in \Sigma} \left\{ z_{\sigma}^{*} \right\} \times \left\{ \sigma \right\} \right) \times [0,N_{0}]$
in its interior, $\Omega_{0}(K) = S$.
\end{prop}
This result follows from the following Lemmas. 
\begin{lemm}
\label{lem1}
If Assumption \ref{assume:2} holds then,
for each compact set $K \subset \Re^{n+2}$,
$\Omega_{0}(K) \subset S$.
\end{lemm}
\begin{lemm}
\label{lem2}
If Assumptions \ref{assume:2} and \ref{assume:3} hold then,
for each compact set $K \subset \Re^{n+2}$ containing the set 
\begin{align}
K_0 := \left( \bigcup_{\sigma \in \Sigma} \left\{ z^{*}_{\sigma} \right\} \times \left\{ \sigma \right\} \right) \times [0,N_{0}]
\end{align}
in its interior, $S \subset \Omega_{0}(K)$.
\end{lemm}
The proof of Lemmas \ref{lem1} and \ref{lem2} follows the same lines as the proof for unconstrained differential equations given in \cite[Section 4.3]{baradaranIFAC2020}. Since the required changes to the proof are minimal
the details are omitted. 
We can finally state our extension of the main results according to \cite{baradaranIFAC2020}. 
\begin{thm}
\label{thmmain}
Under Assumptions \ref{assume:2} and \ref{assume:3}, the set $S$ defined in (\ref{eq:setS}) is semi-globally, practically asymptotically stable in the parameter $\delta >0$ for the system $\hat{\mathcal{H}}_{\delta}$. 
\end{thm}
\subsection{Verifying that the systems in Section III satisfy the assumptions of Theorem \ref{thmmain}}
\begin{thm}
\label{thm:hihbm_assumptions}
Under Assumptions~\ref{objective_assumptions} and \ref{new_objective_assumptions}, and for $\sigma \in \Sigma$, HiHBM with objective function $\phi_{\sigma}$ satisfies Assumption \ref{assume:3}.
\end{thm}
\begin{proof} 
Under the given assumptions, for each constant $\sigma$, let $z^{*}_{\sigma} \coloneqq (q^{*}_{\sigma}, 0)$, where $q^{*}_{\sigma}$ is the unique solution of (\ref{optss}) be the initial value to $\dot{z} \in -F_{\sigma}(z)$. Let $z$ be also a solution to (\ref{sup}). There exist $L \geq 0$ and $r>0$ such that $|F_{\sigma}(z) - F_{\sigma}(z_{\sigma}^{*})| \leq L|z - z_{\sigma}^{*}|$ for all $z \in \left\{ z_{\sigma}^{*} \right\} + r \Ball$. Then, with $e:=z-z_{\sigma}^{*}$ and noting that $e(0)=0$ and $\langle e, \dot{e} \rangle \leq L |e|^{2}$, it follows from standard comparison theorems that $e(t)=0$ for all $t \geq 0$.
\end{proof}
\vspace{-9pt}
\begin{thm}
\label{thm:hihbm_assu}
Under Assumptions~\ref{objective_assumptions} and \ref{new_objective_assumptions}, and for $\sigma \in \Sigma$, HiHBM with objective function $\phi_{\sigma}$ satisfies Assumption \ref{assume:2}.
\end{thm}
\begin{proof}
For each constant $\sigma \in \Sigma$, let $q_{\sigma}^*$ denote the solution to problem \eqref{optss}. Then, global asymptotic stability of $z_{\sigma}^* \coloneqq (q_{\sigma}^*, 0)$ for HiHBM with objective $\phi_{\sigma}$ follows from Theorem~\ref{thm:gas_hihbm} by setting $\mathcal{Q} = \{q_{\sigma}^*\} \times \{0\}$ for each $\sigma$. 
\end{proof}
\begin{pr2} 
\normalfont We now prove Theorem \ref{thm1}. Suppose Assumptions \ref{objective_assumptions} and \ref{new_objective_assumptions} hold. It follows from Theorems \ref{thm:hihbm_assumptions} and \ref{thm:hihbm_assu} that Assumptions \ref{assume:2} and \ref{assume:3} hold. From Theorem \ref{thmmain}, it follows that the set $S$ in (\ref{eq:setS}) is semi-globally practically asymptotically stable for the system (\ref{eq:Hdelta})-(\ref{eq:problem_unswitched}) embedded with an average dwell-time automaton as in (\ref{eq:Hdelta}) with respect to the parameter $\delta$. Finally, according to Proposition \ref{prop11}, the set $S$ is the $\Omega$-limit set indicated in Theorem \ref{thm1}.
\end{pr2}
\vspace{-8pt}
\section{Conclusions}
We discussed the importance of enabling persistent switches of objective functions during an online execution of an optimization algorithm. The reason for these persistent switches is motivated by many engineering applications of convex optimization. We further extend the existing results of \cite{baradaranIFAC2020} on differential equations with multiple equilibria to differential inclusions with distinct multiple equilibria. We present the asset switches occurring in an optimization problem consisting of drones in a relay aiming for maximizing the signal strength. The asset switches are analyzed through the use of average dwell-time parameter which determines the rate of objective's switching during the online HiHBM optimization. Hence, we establish semi-global practical asymptotic stability of a certain set with respect to this parameter. We characterize this set via Omega-limit sets. 
\bibliographystyle{IEEEtran}
\bibliography{ACCbiblio}    
\end{document}